\documentclass[ASNA,twocolumn]{USG} 
\usepackage{anyfontsize} %
\allowdisplaybreaks
\usepackage{colortbl}   
\usepackage{xcolor}     
\usepackage{acro}
\usepackage{cases}
\usepackage{empheq}
\usepackage{steinmetz}
\usepackage{enumitem}
\usepackage{bm}
\graphicspath{{./images/}}
\DeclareAcronym{DFS}{short = DFS, long = Depth-first search}
\DeclareAcronym{SEC}{short = SEC, long = subtour elimination constraints}

\DeclareAcronym{MILP}{short = MILP, long = mixed-integer linear programming}

\DeclareAcronym{MAS}{short = MAS, long = multi-agent systems}

\DeclareAcronym{MINLP}{short = MINLP, long = mixed-integer nonlinear programming}

\DeclareAcronym{MIQP}{short = MIQP, long = mixed-integer quadratic programs}

\DeclareAcronym{MIP}{short = MIP, long = mixed-integer programming}

\DeclareAcronym{MPC}{short = MPC, long = Model predictive control}

\usepackage{makecell}
\articletype{RESEARCH ARTICLE}%
\usepackage{parskip}
\begin{document}
\title{Mixed-integer flow formulations for motion planning and decision-making of networked multi-agent systems}
\transtitle{Guidelines for Establishing a Cytometry Laboratory}
\subtranstitle{trans-subtitle}
\author[1]{Angelo Caregnato-Neto}[https://orcid.org/0000-0002-2401-3506]
\author[2]{Paul-Louis Delacour}[https://orcid.org/0009-0006-2191-6589]
\author[2,3]{Raf Van de Plas}[https://orcid.org/0000-0002-2232-7130]
\author[2]{Tamás Keviczky}[https://orcid.org/0000-0002-2428-2300]
\author[1]{Janito Vaqueiro Ferreira}[https://orcid.org/0000-0001-7324-3482]

\authormark{Caregnato-Neto \textsc{et al.}}
\titlemark{Multi-commodity flow formulations for mixed-integer motion planning of networked multi-agent systems}

\address[1]{\orgdiv{Department of Computational Mechanics (DMC), School of Mechanical Engineering (FEM), }\orgname{State University of Campinas (UNICAMP), }\orgaddress{\state{São Paulo, }\country{Brazil}}}

\address[2]{\orgdiv{Delft Center for Systems and Control, }\orgname{Delft University of Technology, }%
\orgaddress{\state{Delft, }\country{Netherlands}}}

\address[3]{\orgdiv{Dept. of Biochemistry, }\orgname{Vanderbilt University, }%
\orgaddress{\state{ Nashville, TN, }\country{USA}}}

\corres{Angelo Caregnato-Neto  (\email{caregnato.neto@ieee.org})}

\keywords{Connectivity maintenance | Mixed-integer programming | Trajectory planning | Task allocation | Multi-commodity flow}

\abstract[ABSTRACT]{
This work investigates the use of flow-based connectivity maintenance constraints in \ac{MILP} trajectory planning and decision-making models for networked \ac{MAS}.
We integrate flow-based encodings for standard and $k$-hop connectivity into MILP multi-vehicle maneuvering models that are widely used alongside receding horizon planning strategies. Their necessity and sufficiency is demonstrated, guaranteeing full coverage of potential network topologies.
The flow formulation for standard connectivity decreases the growth of the required inequality constraints from exponential to polynomial w.r.t. the size of the MAS when compared to the state-of-the-art subtour elimination (SEC) method.  
The flow-based $k$-hop connectivity constraints decrease the number of required binary variables and decouple its growth from the number of hops.
However, the impact of these formulations in performance is not straightforward due to the introduction of a substantial number of continuous flow optimization variables and, in the case of $k$-hop connectivity, additional inequality constraints. We investigate this trade-off through a statistical evaluation of costs and optimization times using a conventional branch-and-bound commercial solver and trials performed with randomized environments for increasingly larger MAS. The results show that the flow formulation outperforms SEC in standard connectivity problems, enabling the solutions to be computed for larger MAS considering the imposed optimization time limit. 
The reduction in number of binary variables enabled by the $k$-hop flow formulations decreases the theoretical worst-case number of iterations required by the branch-and-bound algorithm to compute the global optimal solution. Our results show that this advantage did not translate into improvements in the average performance when compared to the baseline. 
Nevertheless, we discuss problem settings in which the reduction in binary variables enabled by the flow $k$-hop formulation remains advantageous, particularly for solution strategies that decouple binary and continuous variables such as learning-based MILP solvers.
}

\maketitle

\section{Introduction}\label{sec1}

\Acf{MAS} are systems comprised of multiple autonomous units generally required to collaborate towards a common global goal by pursuing individual objectives. 
The production of simple mobile agents, such as differential-drive robots and multirotors, has become increasingly accessible and cost-effective in recent decades, creating a demand for algorithms that facilitate the deployment of these platforms as collaborative \ac{MAS}, substantially broadening their scope of applications. A thorough analysis of the potential impact of \ac{MAS} on society can be found in \cite{timeline_MAS}. 

The distribution of tasks among several simpler subsystems rather than a single complex agent yields a different design paradigm for \ac{MAS} \cite{Rezaee}.
Their capability to address challenging problems emerges from a collective intelligence, which manifests through the harmonious coordination and collaboration among its units \cite{intro_MAS_book}. In practice, this is accomplished by algorithms that enable MAS to autonomously plan collision-free trajectories and efficiently assign subproblems to individual agents \cite{cena}.

Although not strictly necessary, explicit communication is a central facilitator for the emergence of a collective intelligence in \ac{MAS} \cite{prorok} and has been achieved through a variety of technologies that operate under distinct requirements. For example, omnidirectional  Wi-Fi \cite{robotarium} and Bluetooth \cite{bluetooth} demand distinct conditions for connections to be formed from the directional counterparts, such as infrared and visible light communication \cite{VLC_RCCS}. Furthermore, the maintenance of a local communication network exclusive to the \ac{MAS} offers the benefits of enhanced security, reliability, and dedicated bandwidth over networks shared with other users.  

The allocation and planning requirements of \ac{MAS} can be naturally posed as optimization problems \cite{dmilp_tutorial}. Moreover, the dynamic environments over which mobile agents are typically required to operate also demand the ability to adapt a prior plan to novel and unexpected circumstances, such as the failure of an agent \cite{jcae_resilient} or the emergence of new obstacles and dynamic objectives \cite{moving_target}. 
\Acf{MPC} is an optimization-based control strategy that employs the concept of receding horizon to provide the required robustness. Its suitability for MAS problems has been extensively investigated \cite{survey_schutter} in the context of applications such as autonomous driving \cite{auto_driving}, autonomous coverage \cite{coverage}, target capturing \cite{target_capturing}, and persistent monitoring \cite{persistent_monitoring}. The flexibility of the MPC framework provides diverse formulations for decision-making \cite{kohler_coop,Afonso2020} and collision avoidance \cite{raffo2021,kohler_colavoid}. In the context of the discussed \ac{MAS} problems, a particularly suitable formulation is \ac{MPC} with \acf{MIP} encoding \cite{mip_tutorial}.
MIP is a class of optimization problems characterized by the ability to handle both continuous-valued and, more significantly, integer decision variables. As such, \ac{MIP} offers a powerful framework to model several of the aforementioned complex \ac{MAS} requirements within a unified formulation, providing the benefit of a solution that takes into account all of the \ac{MAS} objectives and constraints simultaneously.
The subclass of \acp{MIP} referred to as convex, which yield convex problems under integrality relaxation and are exemplified by \acf{MILP} and \acf{MIQP}, offers attractive properties. Despite the NP-hardness of general \ac{MIP}, the convex subclass provides a guarantee of convergence to global optimal solutions in a finite number of iterations \cite{mip_tutorial}, while still being capable of modeling notoriously challenging nonconvex problems such as inter-agent collision avoidance, as well as task allocation, and connectivity maintenance \cite{Afonso2020}. Modern solvers \cite{gurobi} can efficiently solve convex \acp{MIP}, enabling integration with the receding horizon strategy of MPC-\ac{MIP} for modest-sized \ac{MAS}. Nevertheless, the scalability of this method w.r.t. number of agents remains as one of the main challenges of the field \cite{prodan_survey}.

 In this context, this work presents a methodology for designing connectivity maintenance constraints for \ac{MILP} models. We leverage the concept of commodity flow \cite{bookflow}, which has been widely used in the field of operations research and network design, but has yet to be fully exploited in the development of trajectory and task allocation frameworks for \ac{MAS}.

\subsection{Related work}

Significant progress has been made on the subject of connectivity maintenance for MAS in the last decades. The problem has been studied in the context of input-to-state stability analysis in \cite{gasparri} and  nonlinear feedback laws for consensus and formation control \cite{egerstedt}. 
Distributed control protocols have been proposed leveraging algebraic connectivity \cite{sabattini} and the ensuing control effort of the maneuvers is decreased by limiting control actions to critical agents, i.e., those whose the removal results in loss of connectivity.
Robustness to model uncertainty has also been addressed with distributed schemes leveraging nonlinear \ac{MPC} \cite{filotheou}. 
The problem of maximizing algebraic connectivity, improving the robustness of a network, has also been investigated using semi-definite programming and distributed optimization in \cite{simonetto}. 
A distributed nonlinear connectivity constraint over the Laplacian matrix of an undirected graph is proposed in \cite{carron} and solved via sequential quadratic programming. 
A consensus-based control method for resilient connectivity maintenance is proposed in \cite{griparic}, with the topology of the network being selected considering a probabilistic model contingent on estimated connectivity levels.
Attention has also been devoted to the issue of strong connectivity maintenance over directed networks, with nonlinear control laws being proposed in \cite{dir_net} for agents under single integrator dynamics.

Sampling methods have also been leveraged to compute trajectories for \ac{MAS} under connectivity requirements. For example, a modified formulation of the rapid-exploring random trees (RRT) method is proposed in \cite{rrt} to guarantee line-of-sight dependent communication between agents, whereas \cite{prm} proposed a scheme based on probabilistic roadmaps (PRM) to dispatch unmanned aerial vehicles to serve as relays in a network connecting a base station and end users.

\begin{figure*}[!ht]
	\centering
	\includegraphics[width=1\textwidth]{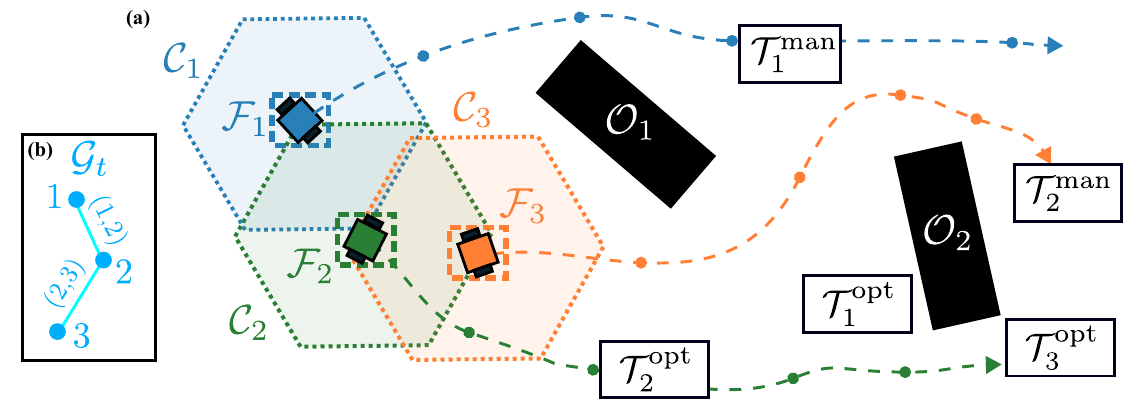}
	\caption{(a) Illustration of motion planning and target allocation problem with three agents, two obstacles, two mandatory, and three optional targets. (b) The undirected graph at time step $t$ representing the communication network under the proximity conditions described by regions $\mathcal{C}_1$, $\mathcal{C}_2$, and $\mathcal{C}_3$. }
	\label{fig:problem}
\end{figure*}

Several MIP-based models have been proposed to address connectivity maintenance of MAS.
In \cite{how_LOS}, the problem of path planning under LOS connectivity requirements is addressed with an MIQP model, where an autonomous helicopter must perform a mission and the remaining aircraft serve as relays connecting it to a ground station. In a similar problem of connecting nodes under LOS requirements, \cite{grotli_milp_conn} proposes a hierarchical architecture where an intermediate \ac{MILP}-based layer plans trajectories to form chain connections. LOS-constrained \ac{MILP} path planning was also studied in the context of pursuer-evader problems in \cite{LOS_evade_purs}. 
\Ac{MINLP} is leveraged in \cite{conn_MINLP} to encode constraints based on realistic stochastic physical layer communication models, ensuring connectivity in a multi-vehicle path planning problem considering predetermined routes. The same technique was employed in \cite{MINLP_signal_degrad} to enforce acoustic connectivity for a group of underwater vessels. 
A \ac{MILP} model to simultaneously address collision-free motion planning, task allocation, and connectivity maintenance of a \ac{MAS} under undirected networks was proposed in \cite{Afonso2020}, with the latter challenge being addressed by \acl{SEC}. Subsequent work introduced formulations for resilient robust connectivity \cite{jcae_resilient}, $k$-hop connectivity \cite{hops}, and connectivity over robot chain control systems considering directional communication technologies \cite{cones}.
The concept of flow has been traditionally used for path planning by roadmap generation over undirected graphs \cite{lasalle_flow}, with recent results being applied to a group of 32 small quadcopters \cite{flow_path}. A network design problem for maximum flow throughput was addressed in \cite{how_flow} using MILP to generate backbone networks considering static and mobile nodes.

\subsection{Contributions}

This work proposes a unified \ac{MILP} model for task allocation and motion planning over a continuous space for networked \ac{MAS} based on \cite{Afonso2020}, where connectivity maintenance is addressed by two commodity flow formulations related to standard and $k$-hop connectivity. 
We establish the necessity and sufficiency of these formulations and contrast their performance and complexity against prior solutions in \cite{Afonso2020} and \cite{hops}. The main contributions are summarized as:
\begin{enumerate}
    \item An MILP flow formulation for motion planning and task allocation under \textit{standard connectivity} which decreases the growth of the number of inequality constraints from exponential \cite{Afonso2020} to quadratic  w.r.t. the number of agents;
    \item An MILP flow formulation for motion planning and task allocation under \textit{$k$-hop connectivity} which decouples the required number of binary optimization variables from the number of hops $k$ and decreases their growth from a fourth-order polynomial \cite{hops} to a quadratic one w.r.t. the number of agents;
    \item A statistical evaluation of the complexity trade-offs between the proposed flow models and their corresponding baselines in \cite{Afonso2020,hops}, assessing their scalability for increasingly larger networked MAS and suitability for receding-horizon planning.
\end{enumerate}

\subsection{Notation}

Define the set of integers in the interval $[a,b]$ as $\mathcal{I}_a^b \triangleq \{a,a+1,\dots, b-1,b\}$ and $\mathbf{1}_n \triangleq [1,1,\dots,1]^\top \in \mathbb{R}^n$. The sets of natural numbers with and without zero are respectively written as $\mathbb{N}$ and $\mathbb{N}^*$. The symbol $\oplus$ denotes the Minkowski sum operation. 
A binary implication $b \implies '\circ'$, $b \in \{0,1\}$, is equivalent to $b = 1 \implies '\circ'$.

\section{Problem statement}\label{sec2}

We consider the problem of planning trajectories and allocating tasks for a group of $n_a \in \{ \mathbb{N}\ \vert\ n_a \geq 2\}$, agents operating in an environment with $n_d \in \{2,3\}$ spatial dimensions. The following assumptions are made.

\begin{assumption}
    The agents operate in regimes where highly aggressive or evasive maneuvers are not required. Under these operating conditions, the non-linear kinematics of the mobile platforms can be adequately approximated or linearized by low-level controllers under linear discrete-time state-space models.
\end{assumption}

\begin{assumption}
    A global localization system, e.g., a global positioning system for outdoor environments or a motion capture system for indoor setups, is available, such that all agents share a common inertial frame of reference, allowing relative distances to be computed centrally.
\end{assumption}

The dynamics and kinematics are represented by linear discrete-time state space models
\begin{align*}
	& \mathbf{x}_{i,t+1} = \mathbf{A}_i \mathbf{x}_{i,t} + \mathbf{B}_i \mathbf{u}_{i,t},\ \mathbf{y}_{i,t} = \mathbf{C}_i \mathbf{x}_{i,t},
\end{align*}
where $\mathbf{x}_{i,t} \in \mathbb{R}^{n_{x,i}}$, $\mathbf{u}_{i,t} \in \mathbb{R}^{n_{u,i}}$ are the state and input vectors of agent $i$ at time step $t$, respectively. The output vector $\mathbf{y}_{i,t} \in \mathbb{R}^{n_d}$ represents the position of the agent w.r.t. an inertial frame of reference. State, input, and output matrices are denoted by $\mathbf{A}_i \in \mathbb{R}^{n_{x,i} \times n_{x,i}}$, $\mathbf{B}_i \in \mathbb{R}^{n_{x,i} \times n_{u,i}}$, and $\mathbf{C}_i \in \mathbb{R}^{n_{d} \times n_{x,i}}$, respectively. The physical limitations of each agent are enforced by bounds on their states $\mathbf{x}_{i,t} \in \mathcal{X}_i$ and inputs $\mathbf{u}_{i,t} \in \mathcal{U}_i$, where $\mathcal{X}_i \subset \mathbb{R}^{n_{x,i}}$ and $\mathcal{U}_i \subset \mathbb{R}^{n_{u,i}}$ are polytopes representing the state and input sets of agent $i$, respectively.

The objective of the group is to reach $n_t \in \mathbb{N}^*$ targets (see Figure \ref{fig:problem}) represented by the polytopes $\mathcal{T}_v \subset \mathbb{R}^{n_d}$, $\forall v \in \mathcal{I}_1^{n_t}$.  
We consider both the cases where targets are assigned \textit{a priori} and their allocation is performed by the MAS. 
In the latter case, the targets are divided into mandatory $\mathcal{T}^\text{man} \subseteq \mathcal{I}_1^{n_t}$ or optional $\mathcal{T}^\text{opt}\subset \mathcal{I}_1^{n_t}$, $\vert \mathcal{T}^\text{man}\vert + \vert \mathcal{T}^\text{opt} \vert = n_t$, $\vert \mathcal{T}^\text{man}\vert \geq 1$. The group must determine the assignments based on a compromise between collecting optional target rewards $r_v \in \mathbb{R}^+,\ v \in \mathcal{T}^\text{opt}$, and performing an economic maneuver in terms of fuel expense and completion time. The mission is complete when all mandatory targets are visited. 
The environment contains $n_o \in \mathbb{N}$ obstacles represented by the polytopes $\mathcal{O}_c \subset \mathbb{R}^{n_d}$, $\forall c \in \mathcal{I}_1^{n_o}$. Inter-agent collisions must also be avoided considering the polytopic approximation of the footprint of each agent, $\mathcal{F}_i \subset \mathbb{R}^{n_d},\ \forall i \in \mathcal{I}_1^{n_a}$.

The group is required to maneuver while preserving the connectivity of a local communication network represented by the time-varying digraph $\mathcal{G}_t = (\mathcal{V},\mathcal{A}_t)$, where $\mathcal{V} = \mathcal{I}_1^{n_a}$ is the set of $n_a$ vertices, each representing an agent, and $\mathcal{A}_t \subseteq \{(i,j) \in\mathcal{V} \times \mathcal{V}\ \vert \ i\neq j \} $ is a time-varying set of arcs corresponding to directed connections formed between them.
The existence of an arc is conditioned on agent $j$ being within transmission range of agent $i$, e.g., $(i,j) \in \mathcal{A}_t \implies \Delta \mathbf{y}_{j,i,t} \in \mathcal{C}_{i}$, where $\Delta \mathbf{y}_{j,i,t} \triangleq \mathbf{y}_{j,t} - \mathbf{y}_{i,t}$ and $\mathcal{C}_{i} = \{\boldsymbol{\sigma} \in \mathbb{R}^{n_d}\ \vert \ \mathbf{P}^\text{c}_{i} \boldsymbol{\sigma} \leq \mathbf{q}^\text{c}_{i} \}$ represents a polytope approximating the connectivity region of agent $i$, where $\mathbf{P}^\text{c}_{i} \in \mathbb{R}^{n_{s,i} \times n_d}$, $\mathbf{q}^\text{c}_{i} \in \mathbb{R}^{n_{s,i}}$, with $n_{s,i} \in \{\mathbb{N}\ \vert\ n_{s,i} \geq 3\}$ denoting the number of hyperplanes defining the corresponding connectivity region. 
\begin{remark}
    In this work, only undirected graphs (referred to simply as a \textit{graph}) are considered. An undirected graph is a special case of a digraph where connections are bidirectional, i.e., $(i,j) \in \mathcal{A}_t \iff (j,i) \in \mathcal{A}_t$.
\end{remark}

Let  $\mathbf{b}^\text{hor} = \left [b^\text{hor}_0,\dots,b^\text{hor}_{\bar{N}} \right]^\top$ denote the vector of time-indexed binary variables $b_t^\text{hor} \in \{0,1\}$ related to a variable planning horizon $N = [0,1,\dots, \bar{N}]\mathbf{b}^\text{hor}$, with $\bar{N} \in \mathbb{N}^*$ being the maximum planning horizon allowed. We develop our study in the context of the following \acf{MILP} model originally proposed in \cite{Afonso2020}.
\begin{subequations}
	\begin{flalign}
		&\text{\textit{Networked \ac{MAS} task allocation and motion planning \ac{MILP} model}} \nonumber \\
		&\underset{ \{\mathbf{u}_{i,t},b_{i,t,v}^\text{tar},b^\text{con}_{i,j,t},b^\text{hor}_t\}  }
		{\textrm{min}}\ N + \sum_{i=1}^{n_a}\sum_{t=0}^{N-1} \gamma_i \Vert \mathbf{u}_{i,t} \Vert_1 - \sum_{t=0}^N \sum_{i=1}^{n_a} \sum_{v \in \mathcal{T}^\text{opt}} r_v b^\text{tar}_{i,t,v} \ \label{cost} \\
		&\textrm{subject to:} \nonumber \\
		& \text{\textbf{Dynamics, input and state bounds}} \nonumber \\
        &\mathbf{x}_{i,0} = \mathbf{x}_{i,\text{init}},\ \forall i \in \mathcal{I}_1^{n_a},\label{const:ini_cond}\\
		&  \mathbf{x}_{i,t+1} = \mathbf{A}_i \mathbf{x}_{i,t} + \mathbf{B}_i \mathbf{u}_{i,t},\ \forall t \in \mathcal{I}_0^{N-1}, \label{const:dyn1} \\
		& \mathbf{y}_{i,t} = \mathbf{C}_i \mathbf{x}_{i,t}, \forall t \in \mathcal{I}_0^{N}, \label{const:dyn2}  \\
		& \mathbf{x}_{i,t} \in \mathcal{X}_i,\ \forall i \in \mathcal{I}_1^{n_a},\forall t \in \mathcal{I}_0^{N}, \label{const:x_bound} \\
		& \mathbf{u}_{i,t} \in \mathcal{U}_i,\ \forall i \in \mathcal{I}_1^{n_a},\forall t \in \mathcal{I}_0^{N-1}, \label{const:u_bound} \\
		&\text{\textbf{Obstacle and inter-agent collision avoidance}} \nonumber \\
		& \mathbf{y}_{i,t} \not \in \mathcal{O}_c,\ \forall c \in \mathcal{I}_1^{n_o},\ \forall t \in \mathcal{I}_0^N, \label{const:obs_avoid} \\
		& \Delta \mathbf{y}_{i,j,t} \not \in \mathcal{F}_{i} \oplus \mathcal{F}_{j},\ \forall i \in \mathcal{I}_1^{n_a-1},\ \forall j \in \mathcal{I}_{i+1}^{n_a}\ \forall t \in \mathcal{I}_0^N, \label{const:col_avoid} \\
		&\text{\textbf{Target allocation and mission completion}} \nonumber \\
		& b^\text{tar}_{i,t,v} \implies \mathbf{y}_{i,t} \in \mathcal{T}_{v},\ \forall v \in \mathcal{I}_1^{n_t},\forall t \in \mathcal{I}_0^{N}, \label{const:tar_entry}  \\
		& \sum_{t=0}^{N} \sum_{i=1}^{n_a} b^\text{tar}_{i,t,v} \leq 1,\ \forall v \in \mathcal{I}_1^{n_t}, \label{const:limit_collect} \\
		& b^\text{hor}_t \leq  \sum_{i=1}^{n_a}\sum_{n=0}^{t} b^\text{tar}_{i,n,v},\  \forall v \in \mathcal{T}^\text{man}, \forall t \in \mathcal{I}_0^{N}, \label{const:terminal} \\
		& \mathbf{1}^\top_{\bar{N}} \mathbf{b}^\text{hor} = 1, \label{const:hor} \\
		&\text{\textbf{Connectivity}} \nonumber \\
		& \mathbf{P}^\text{c}_{i} \Delta \mathbf{y}_{j,i,t} \leq \mathbf{q}_{i}^\text{c} + \mathbf{1}_{n_{s,i}} \left( 1-b^\text{con}_{j,i,t} \right )M,\ (i,j) \in\mathcal{V} \times \mathcal{V},\ i\neq j, \forall t \in \mathcal{I}_0^{N} \label{const:prox} \\
		& \mathbf{H} \mathbf{b}^\text{con} \leq \mathbf{c} \label{const:conn}
	\end{flalign}
\end{subequations}

Cost (\ref{cost}) minimizes a trade-off between time, represented by the optimal horizon $N$; control effort, represented by the accumulated absolute value of the inputs of all agents over $N-1$ time steps and weighted by $\gamma_i \in \mathbb{R}^+,\ \forall i \in \mathcal{I}_1^{n_a}$; and reward collection. 
Constraint (\ref{const:ini_cond}) initializes the prediction model by enforcing that the first state in the planning horizon $\mathbf{x}_{i,0}$ matches the current state of agent $i$, $\mathbf{x}_{i,\text{init}} \in \mathbb{R}^{n_{x,i}}$.
Constraints (\ref{const:dyn1}-\ref{const:dyn2}) enforce the dynamics whereas \eqref{const:x_bound}, (\ref{const:u_bound}) bound state and input, respectively. Obstacle and inter-agent avoidance are enforced by (\ref{const:obs_avoid}) and (\ref{const:col_avoid}), respectively. 
Given the target binary variable $b^\text{tar}_{i,t,v} \in \{0,1\}$, constraint \eqref{const:tar_entry} encodes the conditions for target visitation, while \eqref{const:limit_collect} prevents agents from collecting rewards from the same target multiple times. Mission completion is enforced by (\ref{const:terminal}) which conditions the activation of the horizon binary on the visitation of all mandatory targets, whereas (\ref{const:hor}) guarantees that the horizon binary is activated strictly once in at most $\bar{N}$ time steps.  

We note that the compact notation employs a cost and constraints (\ref{const:dyn1}-\ref{const:hor})  enforced for $0 \leq t \leq N$, i.e., up to the optimal horizon. Since this quantity cannot be known \textit{a priori}, in practice, they are written explicitly for $0 \leq t \leq \bar{N}$ and relaxed for $t \geq N$. 
Since this work is concerned solely with connectivity, we present only the relevant constraints in their explicit form for the sake of brevity. The reader is referred to \cite{Afonso2020,varHor} for a detailed description of (\ref{const:dyn1}-\ref{const:hor}).

Constraint (\ref{const:prox}) encodes the proximity conditions for the existence of an arc $(i,j)$, i.e., $b^\text{con}_{j,i,t} \implies \Delta \mathbf{y}_{j,i,t} \in \mathcal{C}_i$, where $b^\text{con}_{j,i,t} \in \{0,1\}$ is a connectivity binary variable. This is accomplished using the ``Big-M" variable $M > 0$, which relaxes (\ref{const:prox}) if the corresponding binary is equal to zero. Finally, \eqref{const:conn} encodes conditions on the vector of all connectivity binary variables $\mathbf{b}^\text{con} \in \{0,1\}^{n_c}$, $n_c = n_a(n_a-1)(\bar{N}+1)$, $\mathbf{H} \in \mathbb{R}^{n_c\times n_c}$, $\mathbf{c} \in \mathbb{R}^{n_c}$, such that a connected communication network can be constructed. 

\section{MIP connectivity formulations}

\subsection{Standard connectivity}\label{sec:subtour}
A graph $\mathcal{G}_t = (\mathcal{V},\mathcal{A}_t)$ is connected if it contains at least one path between any pair of vertices belonging to $\mathcal{V}$. In this work, this property is referred to as standard connectivity.
In \cite{Afonso2020}, Subtour Elimination Constraints (SEC) are employed to handle the problem of standard connectivity maintenance over MAS networks represented by undirected graphs. 
The solution follows the rationale that the task of the motion planner is to \textit{enable} the construction of a connected network by properly positioning the agents while they spread to reach the targets. The \acp{SEC} satisfy this requirement by guaranteeing that the group's displacement allows the construction of a spanning tree at each time step, which is a necessary and sufficient condition for standard connectivity. 

In the following, we present a \ac{MILP} formulation for standard connectivity based on flows that can be readily integrated into the MAS motion planning and task allocation MILP model as an alternative  to SEC.

\begin{definition}
A time-varying flow network is defined on a digraph $\mathcal{G}_t = (\mathcal{V},\mathcal{A}_t)$, $\forall t \in \mathcal{I}_0^{\bar{N}}$, associated with flow variables $f_{i,j,t}^{(\ell)} \geq 0$, $\forall (i,j) \in \mathcal{A}_t$ where a virtual commodity $\ell$ starting at a source  $s \in \mathcal{V}$ is required to reach a set of sinks $\mathcal{P} \subseteq \mathcal{V}\setminus \{s\}$ while satisfying the following conditions:
\begin{align}
    & \sum_{j \in \mathcal{V}\setminus \{i\}} f_{i,j,t}^{(\ell)} - \sum_{u \in \mathcal{V}\setminus \{i\}} f_{u,i,t}^{(\ell)} = \begin{cases}
    w_\ell, \text{ if } i = s,\\
    -w_{\ell,i}, \text{ if } i \in \mathcal{P},\\
    0, \text{ otherwise},
    \end{cases} \text{ (balance),}\label{def:flow_bal}\\
     & f_{i,j,t}^{(\ell)} \leq c_{i,j,t} \text{ (capacity),} \label{def:flow_cap}
\end{align}
\end{definition}

where $c_{i,j,t} \geq 0$ is the capacity associated with an arc $(i,j) \in \mathcal{A}_t$, $w_\ell = \sum_{i \in \mathcal{P}} w_{\ell,i}$ is the total commodity $\ell$ supplied by the source and $w_{\ell,i} > 0$, $i \in \mathcal{P}$, are the assigned demand for each sink $i \in \mathcal{P}$. 

The flow variables associated with a single commodity $f_{i,j,t} \geq 0,\  (i,j) \in  \mathcal{V} \times \mathcal{V},\ i \neq j$, which are continuous-valued decision variables declared for all potential existing arcs, as $\mathcal{A}_t$ is unknown \textit{a priori}.
 Without loss of generality, we select node $s \in \mathcal{V}$ to be the source, and the remaining nodes to be sinks, $\mathcal{P} =  \mathcal{V}\setminus \{s\}$. In this context, the following constraints (\ref{const:cap}-\ref{const:sink}) encode necessary and sufficient conditions for standard connectivity, $\forall t \in \mathcal{I}_0^{N}$ as proven in theorem~\ref{thm:connectivity}:
\begin{subequations}\label{global}
\begin{align}
    \begin{cases}
        f_{i,j,t} \leq (n_a - 1)b^\text{con}_{i,j,t}\\
        f_{j,i,t} \leq (n_a - 1)b^\text{con}_{i,j,t} 
    \end{cases}, &\  \forall (i,j) \in \mathcal{V}\times\mathcal{V} , \: i<j, \label{const:cap}
\end{align}
\begin{empheq}[
  left={\sum_{j \in \mathcal{V}\setminus\{i\}} (f_{i,j,t} - f_{j,i,t}) =\empheqlbrace}
]{alignat=2}
    &n_a-1           &\qquad& \text{if $i = s$,}    \label{const:source} \\
    &-1 &\qquad& \text{if } i \in \mathcal{V}\setminus\{s\}. \label{const:sink}
\end{empheq}
\end{subequations}

\begin{theorem}
\label{thm:connectivity}
	The graph $\mathcal{G}_t$ is connected if and only if there exists a flow satisfying the constraints (\ref{const:cap}-\ref{const:sink}).
\end{theorem}
\begin{proof}
    \textit{Necessity}. Since $\mathcal{G}_t$ is connected, we can use a \ac{DFS} to extract a spanning tree subgraph and record for every node $u$ the number $n(u)$ of visited nodes in the subtrees. Define $\mathcal{E}$ as the set of edges traversed by the DFS algorithm. Then, the following flow can be constructed:
    \begin{align*}
        f_{u,v,t} &=
        \begin{cases}
            n(v)+1 & \text{if } (u,v) \in \mathcal{E} \text{ and } n(u)>n(v), \\
            0 &\text{else}.
        \end{cases}
    \end{align*}
    Since this flow traverses the spanning tree structure, which is a subgraph of $\mathcal{G}_t$ induced by $b^\text{con}_{i,j,t}$ through \eqref{const:cap}, reaching all sinks from the source, it must also satisfy the balance constraints \eqref{const:source} and \eqref{const:sink}.\\
    \textit{Sufficiency}. We prove it by contradiction and assume that $\mathcal{G}_t$ is disconnected. Let $\mathcal{R} \subseteq \mathcal{V}$ be the set of vertices reachable from the source satisfying: $\vert \mathcal{R} \vert \leq n_a-2$. Summing \eqref{const:source} and \eqref{const:sink} for $i\in \mathcal{R}$, we get :
    \begin{align}
        \sum_{i\in \mathcal{R}} \sum_{j\in \mathcal{V}\setminus \{i\}} f_{i,j,t}-f_{j,i,t} = n_a-1 -\vert \mathcal{R} \vert \geq 1
        \label{eq:reachable-v}
    \end{align}
    In \eqref{eq:reachable-v}, for any $i \in \mathcal{R}$, if $j\in \mathcal{R}$ then the term $f_{i,j,t}$ appears twice, once positively and once negatively, and thus cancels. We thus get that:
    \begin{align*}
        \sum_{i\in \mathcal{R}} \sum_{j\in \mathcal{V}\setminus \mathcal{R}} f_{i,j,t}-f_{j,i,t} \geq 1.
    \end{align*}
    This means that there exists $i\in \mathcal{R}$ and $j\in \mathcal{V}\setminus \mathcal{R}$ such that $f_{i,j,t}>0$. By assumption, no other vertices are reachable from $\mathcal{R}$, meaning that $b_{i,j,t}^{\text{con}} = 0$, which contradicts \eqref{const:cap}.
\end{proof}

As shown in Table \ref{tab:complex_standard}, the number of binary variables introduced by the SEC and flow formulations is the same. The main distinction comes from the substantial decrease in constraints achieved by the flow formulation, which scales quadratically whereas the SEC scales exponentially. This represents a significant decrease in the size of the MILP optimization model. However, the flow formulation requires the introduction of additional continuous variables that grow quadratically with the number of agents.

\begin{table}[ht]
    \centering
    \caption{Complexity of standard connectivity formulations.}
    \label{tab:complex_standard}
    \renewcommand{\arraystretch}{2.0}
    \setlength{\tabcolsep}{4pt}
    \begin{tabular}{lcccc}
        \toprule
        \textbf{Method}
            & \makecell{\textbf{Binary}\\ \textbf{variables}}
            & \makecell{\textbf{Continuous}\\ \textbf{variables}}
            & \textbf{Constraints} 
            & \makecell{\textbf{Necessity \&}\\ \textbf{sufficiency}}\\
        \midrule
        SEC~\cite{Afonso2020}
            & $O\!\left(n_a^2\right)$
            & $0$
            & $O\!\left(2^{n_a}\right)$ 
            & Yes \\
        Flow 
            & $O\!\left(n_a^2\right)$
            & $O\!\left(n_a^2\right)$
            & $\bm{O\!\left(n_a^2\right)}$ 
            & Yes \\
        \bottomrule
    \end{tabular}
\end{table}

\subsection{\textit{k}-hop connectivity}\label{sec:hops_afonso}

We start the discussion regarding $k$-hop connectivity by listing a few fundamental definitions.

\begin{definition}
    A path in graph $\mathcal{G}_t$ is a sequence of vertices $\{v_c\}_{c=1}^n$ where $(v_c,v_{c+1}) \in \mathcal{A}_t$ and $v_i \neq v_j$ for all $i, j \in \{1, \dots, n\}, i \neq j$.
\end{definition}

\begin{definition}
     The \textit{length} of a path is its number of edges $L\left( \{v_c\}_{c=1}^n \right)=n-1$. The number of hops in a path is equal to its length. The \textit{distance} between a pair of vertices is the length of the shortest path between them. The \textit{diameter} of a connected graph, $\text{diam}(\mathcal{G}_t)$, is the maximum distance between any two vertices in $\mathcal{V}$.
\end{definition}

\begin{definition}
     A connected graph is $k$-hop connected if and only if $\text{diam}(\mathcal{G}_t) \le k$.
\end{definition}
The $k$-hop connectivity is a desired property of networks due to the limits it imposes on the distance a message has to traverse to reach its destination, thereby reducing energy consumption  and enhancing communication quality \cite{min_energy_networks}.
Necessary and sufficient conditions for this property have been proposed as linear constraints of MILP models in \cite{hops}. The formulation is constructed using conditions over the power of binary variable matrices, which yield nonlinear constraints that must be linearized through the introduction of auxiliary binary variables and constraints, substantially increasing the theoretical complexity of the optimization model.

Consider the time-indexed graph $\mathcal{G}_t = (\mathcal{V},\mathcal{A}_t)$. Leveraging the concept of layered graphs \cite{khop_flow}, we define the flow variables with an extra dimension representing a layer $h$, $f_{i,j,t,h}^{(\ell)}\geq 0$, $\forall (i,j) \in \mathcal{V}\times\mathcal{V}$, $\forall \ell = (s,p) \in \mathcal{K}$, $\forall t \in \mathcal{I}_0^{\bar{N}}$, $\forall h \in \mathcal{I}_1^k$, where $k \in \mathbb{N},\ 1 \leq k \leq n_a-1$, denotes the \textit{desired} diameter (maximum number of hops) of the network. It follows that each arc in the network is associated with a distinct flow variable for each layer of $\mathcal{G}_t$.
The following multi-commodity flow constraints, based on \cite{khop_flow}, are proposed to enforce $k$-hop connectivity , $\forall t \in \mathcal{I}_0^{\bar{N}}$, $\forall \ell = (s,p)\in \mathcal{K}$,
\begin{align}
    &\begin{cases}
		f_{i,j,t,h}^{(\ell)} \leq b^\text{con}_{i,j,t}\\
		f_{j,i,t,h}^{(\ell)} \leq b^\text{con}_{i,j,t} 
	\end{cases},\  \forall (i,j) \in \mathcal{V}\times\mathcal{V} , \: i<j,\ \forall h \in \mathcal{I}_1^k, \label{const:flow_khop}\\
    & \sum_{j \in \mathcal{V}\setminus \{s\} } f_{s,j,t,1}^{(\ell)} = 1, \label{const:source_khop} \\
    & \sum_{h=2}^{k}\  \sum_{i \in \mathcal{V}\setminus \{p\} } f_{i,p,t,h}^{(\ell)} + f_{s,p,t,1}^{(\ell)} = 1 \label{const:sink_khop}, \\
    &f_{s,i,t,1}^{(\ell)} - \sum_{j \in \mathcal{V}\setminus\{i,s\}} f _{i,j,t,2}^{(\ell)} = 0,\ \forall i \in \mathcal{V}\setminus\{s,p\}, \label{const:relay_khop1} \\
    & \sum_{j \in \mathcal{V}\setminus\{i,p\}} f_{j,i,t,h}^{(\ell)} - \sum_{j \in \mathcal{V}\setminus\{i,s\}} f_{i,j,t,h+1}^{(\ell)} = 0,\ \forall i \in \mathcal{V}\setminus\{s,p\}, \forall h \in \mathcal{I}_2^{k-1} \label{const:relay_khop2} 
\end{align}

 Theorem~\ref{thm:k-hop-connectivity} shows that the flow constraints (\ref{const:flow_khop}–\ref{const:relay_khop2}) constitute necessary and sufficient conditions for k-hop connectivity.

\begin{theorem}
\label{thm:k-hop-connectivity}
	The graph $\mathcal{G}_t$ is $k$-hop connected if and only if there exists a flow satisfying the constraints (\ref{const:flow_khop}-\ref{const:relay_khop2}) for all $(s,p)$.
\end{theorem}

\begin{proof}
    \textit{Necessity}. If $\mathcal{G}_t$ is $k$-hop connected, then for any two vertices $(s,p)\in \mathcal{V}$ there is a path : $s = v_0, \ldots , v_h = p$ with $h\leq k$. Let $f$ be the flow routing one unit along that path:
    \begin{align*}
        f_{v_0,v_1,t,1}^{(\ell)} = 1, f_{v_1,v_2,t,2}^{(\ell)}=1, \ldots, f_{v_{h-1},v_h,t,h}^{(\ell)} = 1,  
    \end{align*}
    and all the other flow variables are set to $0$.
    \begin{itemize}
        \item \eqref{const:flow_khop}: each active flow variable corresponds to an edge. 
        \item \eqref{const:source_khop} is satisfied since:
        \begin{align*}
            \sum_{j \in \mathcal{V}\setminus \{s\} } f_{s,j,t,1}^{(\ell)} = f_{s,v_1,t,1} = 1 
        \end{align*}
        \item The flow arrives at $p=v_h$ at hop $h$.
        If $h=1$, $f_{s,p,t,1}=1$, and otherwise $f_{v_{h-1},p,t,h}=1$. In both cases \eqref{const:sink_khop} is satisfied. 
        \item \eqref{const:relay_khop1}: for any relay node $i \in \mathcal{V}\setminus \{s,p\}$, if $i=v_1$, $f_{s,v_1,t,1}^{(\ell)} = 1$ and $f_{v_1,v_2,t,2}^{(\ell)}=1$. For all the others $i\in \mathcal{V} \setminus \{s,p,v_1\}$, both sides are zero.
        \item \eqref{const:relay_khop2}: for any relay node $i \in \mathcal{V}\setminus \{s,p\}$ and $h \in \mathcal{I}_{2}^{k-1}$, if $i=v_h$, exactly one unit of flow arrives at hop $h$ and leaves at hop $h+1$: $f_{v_{h-1},v_h,h} = f_{v_h,v_{h+1},h+1}$. For all the other nodes, both sides are $0$, and in both cases the constraint is satisfied.
    \end{itemize}
    \textit{Sufficiency}. Let $(s,p) \in \mathcal{K}$ be arbitrary. We show that the existence of a flow satisfying the constraints (\ref{const:flow_khop}-\ref{const:relay_khop2}), ensures that there is a path of at most $k$ hops in $\mathcal{G}_t$ between $s$ and $p$. 
    From equation \eqref{const:sink_khop}, either :
    \begin{itemize}
        \item  $f_{s,p,t,1}>0$, in which case $b_{s,p,t}^{\text{con}}$ must be $1$ from \eqref{const:flow_khop}. In this case, there is a $1$-hop path from $s$ to $p$.
        \item $f_{j_{h-1},p,t,h}^{(\ell)} > 0$ for some $h\geq 2$. Similarly, this implies using \eqref{const:flow_khop} that $(j_{h-1},p)$ is an edge in $\mathcal{G}_t$. 
    \end{itemize}
    We can now trace backward from $j_{h-1}$.
    \begin{itemize}
        \item If $h=2$, applying \eqref{const:relay_khop1} to $i=j_1$ gives $f_{s,j_1,t,1}^{(\ell)}>0$, so $(s,j_1)$ is an edge in $\mathcal{G}_t$, yielding the path $s,j_1,p$ of 2 hops.
        \item If $h\geq 3$, applying \eqref{const:relay_khop2} to $i = j_{h-1}$ yields $j_{h-2} \in \mathcal{V}\setminus\{j_{h-1},p\}$ with $f_{j_{h-2},j_{h-1},t,h-1}^{(\ell)} >0$, and edge $(j_{h-2},j_{h-1})$ in $\mathcal{G}_t$. Continuing backward until $h=2$, we get a sequence of edges $(j_1,j_2),(j_2,j_3),\ldots, (j_{h-1},j_h)$ in the graph $\mathcal{G}_t$ such that $f_{j_{i-1},j_{i},i}^{(\ell)} >0$ for all $i \in \mathcal{I}_2^{h-1}$. Finally, applying \eqref{const:relay_khop1} to $i=j_1$, ensures that $f_{s,j_1,t,1}^{(\ell)} >0$, and thus that the edge $(s,j_1)$ is in $\mathcal{G}_t$, yielding the path $s, j_1, \ldots, j_{h-1} , p$ of $h \leq k$ hops.
    \end{itemize}
\end{proof}

Table \ref{tab:complex_k_hop} shows that the flow formulation for $k$-hop connectivity requires the same number of connectivity binary variables as the previous cases and, contrary to the method presented in \cite{hops}, does not scale with the desired diameter $k$. However, the proposed approach also requires the introduction of a substantial number of continuous variables and inequality constraints, with both quantities scaling with $k$.
\begin{table}[ht]
    \centering
    \caption{Complexity of $k$-hop connectivity formulations.}
    \label{tab:complex_k_hop}
    \renewcommand{\arraystretch}{2.0}
    \setlength{\tabcolsep}{3pt}
    \begin{tabular}{lcccc}
        \toprule
        \textbf{Method}
            & \makecell{\textbf{Binary}\\ \textbf{variables}}
            & \makecell{\textbf{Continuous}\\ \textbf{variables}}
            & \textbf{Constraints} 
            & \makecell{\textbf{Necessity \&}\\ \textbf{sufficiency}}\\
        \midrule
        \makecell[l]{Matrix\\power~\cite{hops}}
            & $O\!\left(n_a^3k\right)$
            & $0$
            & $O\!\left(n_a^3k\right)$
            & Yes \\
        Flow 
            & $\bm{O\!\left(n_a^2\right)}$
            & $O\!\left(n_a^4k\right)$
            & ${O\!\left(n_a^4k\right)}$ 
            & Yes \\
        \bottomrule
    \end{tabular}
\end{table}

\section{Computational evaluation}\label{sec3}

\begin{figure*}[!ht]
	\centering
	\includegraphics[width=0.95\textwidth]{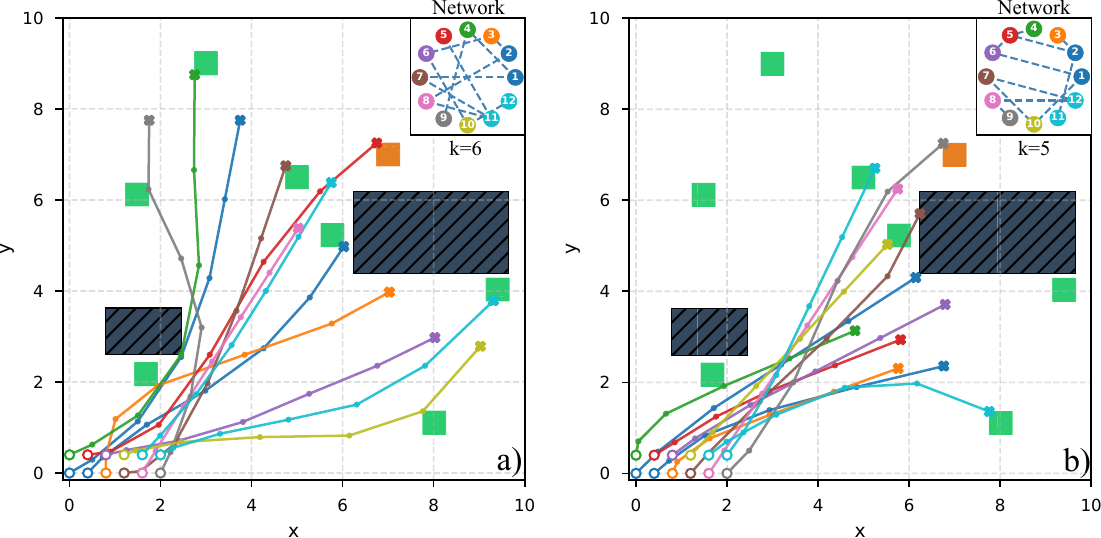}
	\caption{Trajectories and target assignment computed with models under a) flow and b) SEC standard connectivity constraints in the benchmark environment from \cite{Afonso2020}. The graph representing the communication network formed at the last time step is presented in the upper right corner of both figures.}
	\label{fig:stdcon_example}
\end{figure*}

We divide the presentation of the results in two parts. First, an illustrative example comparing the flow and SEC approaches (Figure \ref{fig:stdcon_example}) for standard connectivity. Second, a statistical evaluation over 100 randomized scenarios with increasingly larger MAS (Table \ref{tab:mc_failure} and Figure \ref{fig:results}).

In both cases, we consider a homogeneous MAS with agents' dynamics represented by double integrators (sampling period $T=1$ second). The linear velocities were bounded to the interval $[-10,10]$ and the corresponding accelerations to $[-5,5]$. The connectivity regions were modeled as squares of side 2. The optimization weights were chosen as $\gamma_i = 1, \forall i \in \mathcal{I}_1^{n_a}$ and $r_v = 10, \forall v \in \mathcal{T}^\text{opt}$.  The maximum planning horizon was set as $\bar{N}=12$. All problems were solved using the Gurobi optimizer \cite{gurobi} on a computer with an i5-1135G7 (2.4-4.2GHz) and 16GB RAM. 

\subsection{Networked MAS motion planning example}
To illustrate the functionality of the proposed motion planning and decision-making algorithm, we employ the environment presented in Figure \ref{fig:stdcon_example} for $n_a=12$, which is based on the benchmark established in \cite{Afonso2020}. For this particular example, the optimization time limit was selected as 600 seconds. Figure \ref{fig:stdcon_example}a shows the trajectories and target assignments considering the flow standard connectivity approach. The flow method enabled the MAS to reach each of the 7 optional targets in a maneuver of 6 seconds, whereas in the case of SEC (Figure \ref{fig:stdcon_example}b) only 4 optional targets were visited in a trajectory of 5 seconds. The connected communication network formed at the final step for each method is presented in the upper right corner of the corresponding figures, illustrating how the optimization enables the preservation of the communication by assigning some agents to the task of visiting targets while others serve as network relays. The final  costs of the flow and SEC formulations were -69.9 and -41.7, respectively.

\subsection{Randomized trials}
\begin{figure*}[!ht]
	\centering
	\includegraphics[width=0.98\textwidth]{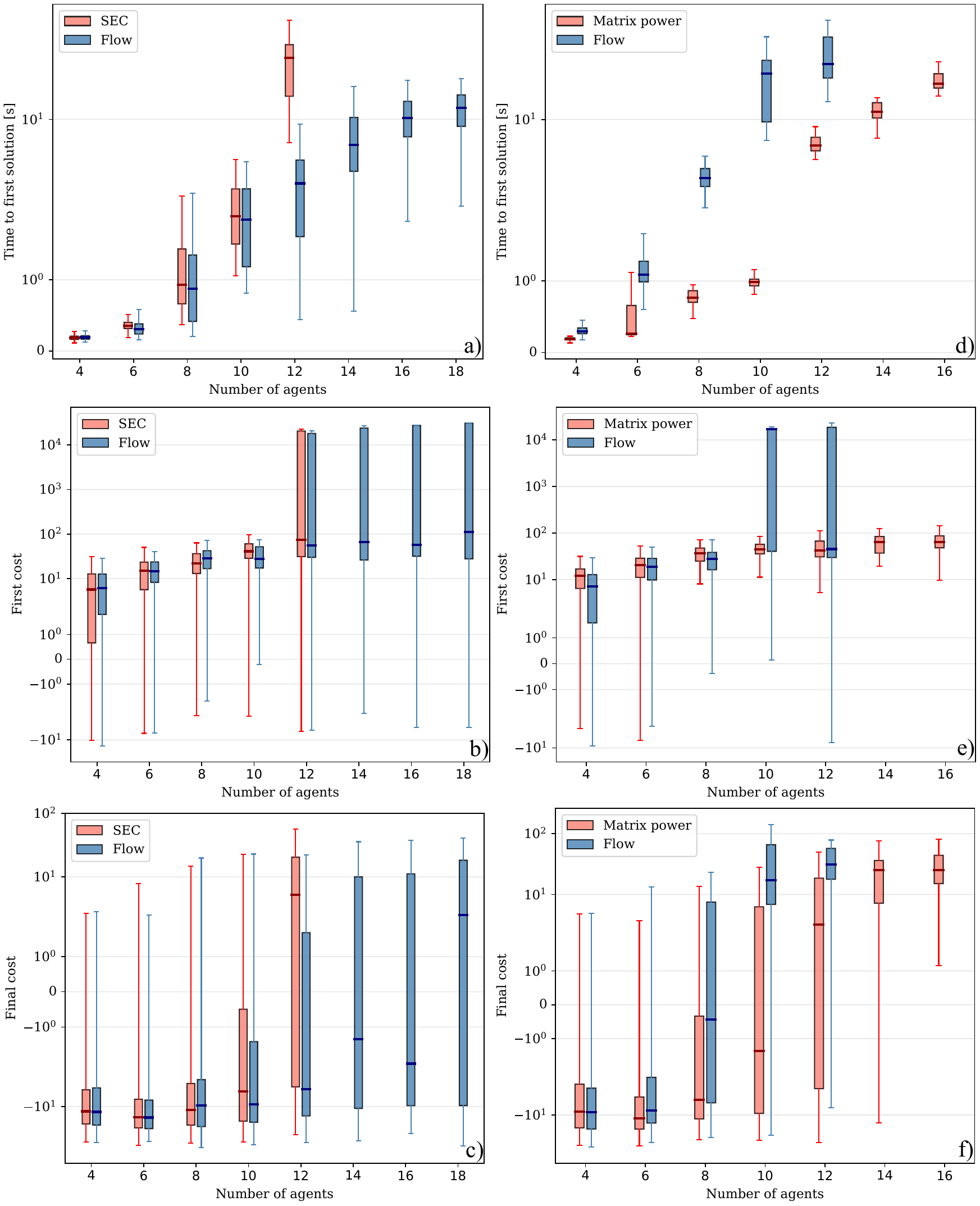}
	\caption{Statistical results of optimization times and solution quality for 100 trials and a time limit of 60 seconds. a) time to first solution; b) cost of first solution; and c) cost of the best solution found within the time limit for standard connectivity. d) time to first solution; e) cost of first solution; and f) cost of the best solution found within the time limit for $3$-hop connectivity. Plots where the rate of failure was above 50\% were omitted.}
	\label{fig:results}
\end{figure*}
We performed a statistical evaluation of the flow methods for standard and $k$-hop connectivity considering randomized environments to compare their performance with the baselines SEC \cite{Afonso2020} and the matrix power method \cite{hops} for optimizations constrained to 60 seconds considering MAS with $n_a \in  \{4,6,8,\dots,20\}$. The performance was assessed in terms of:
\begin{enumerate}[label=\alph*)]
    \item The rate of failure in finding a solution within the established time limit;
    \item The quality of the best solution (if available) found within this time window; 
    \item The time required to compute the first feasible solution;
    \item The first solution's quality.
\end{enumerate}
We note that items c) and d) are relevant metrics for receding horizon methods where feasible solutions must be computed in a relatively small time window.

The environments were generated with $n_o \sim \mathcal{U}(0,2)$, $n_t \sim \mathcal{U}(n_a-2,n_a+2)$. The positions of the obstacles and targets were determined using uniform distributions. Square obstacles with side $\mathcal{U}(0.5,2)$ were considered.  In the $k$-hop connectivity cases, the number of hops was selected as $k=3$.

Table \ref{tab:mc_failure} presents the rate of failures considering the generated scenarios for each method and connectivity type as the number of agents in the MAS increases. We consider a model intractable if the rate of failure in finding solutions within the established time limit is greater than or equal to 50\%. 
The results show that the SEC and flow formulations for standard connectivity yielded tractable models for $n_a \leq 14$ and $n_a < 18$, respectively.
Remarkably, we observe that the rate of failures almost doubles in the transition from 12 to 14 agents considering SEC, whereas the models under flow constraints presented a more controlled growth, indicating better scalability w.r.t. the size of the MAS.
\begin{table}[ht]
    \centering
    \caption{Failure rate (\%) over 100 Monte Carlo trials.}
    \label{tab:mc_failure}
    \renewcommand{\arraystretch}{1.4}
    \setlength{\tabcolsep}{4pt}
    \begin{tabular}{lrrrrrrrrr}
        \toprule
        $n_a$ & 4 & 6 & 8 & 10 & 12 & 14 & 16 & 18 & 20 \\
        \midrule
        \multicolumn{10}{l}{\textit{Standard conn.}} \\
        \quad SEC
            &  0 &  0 &  0 &  0 & 44 & 85 & 100  & 100  & 100  \\
        \quad Flow
            &  0 &  0 &  0 &  2 & 10 & 15 & 35 & 44 & 59 \\
        \midrule
        \multicolumn{10}{l}{\textit{3-hop conn.}} \\
        \quad Matrix power
            &  0 &  0 &  1 &  1 &  15 &  30 &  50 &  -- &  -- \\
        \quad Flow
            &  0 &  0 &  0 &  10 &  48 &  92 &  100 &  -- &  -- \\
        \bottomrule
    \end{tabular}
\end{table}


Figure \ref{fig:results} shows the statistical results. The data related to intractable cases (50\% or more rate of failure) were omitted due to the insufficient samples.
Figures \ref{fig:results}a, \ref{fig:results}b, and \ref{fig:results}c, show that both methods performed similarly in all metrics for $n_a \leq 10$.
However, the SEC formulation takes substantially more time to compute the first feasible solution than the flow for $n_a = 12$, and becomes intractable for larger MAS. The flow formulation consistently provides solutions within the time limit for $n_a \leq 18$. However, the variability in the value of the initial and final costs grows substantially for $n_a \geq 12$. Nevertheless, Figure \ref{fig:results}c demonstrates that the solver is able to bring the final solutions to a cost between $[-100,100]$.

In the case of $3$-hop connectivity, Table \ref{tab:mc_failure} shows that the matrix power and flow formulations provided tractable models for $n_a \leq 16$ and $n_a \leq 12$, respectively.
An approximated five-fold increase in the rate of failure was observed for the flow method in the transition from a model with $n_a=10$ to $n_a = 12$, which became intractable for $n_a=14$. 
The existing matrix power formulation outperformed flow in terms of scalability, depicting a more controlled growth in the number of failures as the size of the MAS increases. 
Figure \ref{fig:results}d illustrates the substantial time required to compute first solution under the flow scheme for $n_a\geq 8$. While the baseline generally performs better, we observe an approximate 6-fold growth in time for $n_a=12$. 
Figures \ref{fig:results}e demonstrate that for $10 \leq n_a \leq 12$ the flow method  provides initial solutions with much more variance, being able to compute better first costs in a few instances, while the baseline consistently computes solutions with cost between around 10 and 100. However, \ref{fig:results}f shows that the matrix power method provides much better capabilities of improving the initial solutions for $n_a \geq 8$.

\subsection{Discussion}

The results demonstrate that the flow encoding for standard connectivity provides substantial improvements in terms of scalability w.r.t. the size of the MAS, enabling the motion planning and decision-making problems to be solved for larger MAS within the established time limit compared to the SEC baseline.
We attribute this result to the overall reduction in the scale of the MILP model provided by the flow formulation. The number of inequality constraints required by the SEC encoding grows exponentially with the number of agents, rendering the model intractable at moderate scales. In practice, the large number of constraints severely increases the time required by the solver's presolve routines, which are a fundamental element of modern solvers, responsible for removing redundancies and tightening bounds. Indeed, for $n_a=14$, the average presolve time under SEC was 28.2 seconds compared to 0.58 seconds under the flow formulation, indicating that nearly half of the allowed optimization time was spent in presolve in the case of SEC.

In the case of $k$-hop connectivity, the main advantage of the proposed flow encoding is a substantial reduction in binary optimization variables, which dictate the worst-case number of iterations required by a branch-and-bound solver to find a global optimal solution. However, this reduction did not translate into practical performance gains within the imposed time limit. The overall increase in model size resulting from the additional continuous variables and inequality constraints resulted in worse scalability of the flow encoding compared to the matrix power baseline. In contrast to the standard connectivity case, the presolve times were comparable between the flow and matrix power formulations for $k$-hop connectivity, suggesting that the performance gap originates from the cost of solving the LP relaxations at each branch-and-bound node. This is consistent with the $O(n_a^4k)$ growth of continuous variables and inequality constraints introduced by the flow encoding, which substantially increases the scale of the corresponding LP relaxations when compared to the matrix power baseline.

\section{Conclusion}\label{sec4}

This work investigated the integration of connectivity maintenance constraints based on the concept of flow into \ac{MILP} models used for motion planning and decision-making of \ac{MAS}. 
We demonstrated that these alternative formulations encode necessary and sufficient conditions for standard and $k$-hop connectivity,
and evaluated their performance  over the baselines using a statistical analysis over randomized trials for MAS of increasing sizes and a standard branch-and-bound commercial solver.

We conclude that the flow formulation is the most promising MILP implementation of standard connectivity constraints for networked MAS under branch-and-bound solve strategies. As such, future work should be directed at exploring methods to improve the scalability of these models, e.g., distributed MILP \cite{dmilp_tutorial}.
Conversely, the flow formulation for $k$-hop connectivity did not provide benefits over the baseline under restricted optimization times. However, we still see value in the reduction of binary variables provided by the flow implementation considering distinct solve strategies. In particular, we believe that methods based on decoupling the solution of binary and continuous variables can benefit substantially from this formulation. For example, deep reinforcement learning can be used to train a model to determine only the binary variables for a mixed-integer MPC controller \cite{caio_paper}. The reduced number of variables enables the use of smaller neural networks that can be efficiently trained and potentially provide better inference performance.

\bmsubsection*{Acknowledgments}
The authors have nothing to report.

\bmsubsection*{Funding}
This study was financed, in part, by the São Paulo Research Foundation (FAPESP), Brazil. Process Numbers 2024/04703-0 and 2025/01170-4.

\bmsubsection*{Disclosure}
The authors have nothing to report.

\bmsubsection*{Conflicts of Interest}

The authors declare no conflicts of interest.

\bmsubsection*{Data availability statement}
The data that support the findings of this study are available from the corresponding author upon reasonable request.

\bibliography{wileyNJD-Chicago}

\nocite{*}

\end{document}